\documentclass[sigconf]{acmart}

\usepackage{booktabs} 
\usepackage[colorinlistoftodos]{todonotes}
\usepackage{algorithmic}
\usepackage[linesnumbered, ruled, vlined]{algorithm2e}
\usepackage{multirow}
\usepackage{tablefootnote}
\usepackage{graphicx}
\usepackage{subfigure}
\usepackage{flushend}
\usepackage{color}
\usepackage{enumitem}
\usepackage{array, makecell}
\usepackage{footmisc}

\usepackage{xcolor}
\newcommand\yang[1]{\textcolor{red}{[Hongxia: #1]}}
\newcommand{\company}{Alibaba}
\newcommand{\model}{{\em GATNE}}
\newcommand{\bi}{base}
\newcommand{\gij}{edge}

\newcommand\green[1]{\textcolor{green}{#1}}
\newcommand{\hide}[1]{} 
\newcommand{\vpara}[1]{\vspace{0.07in}\noindent\textbf{#1 }}

\setcopyright{rightsretained}

\copyrightyear{2019} 
\acmYear{2019} 
\setcopyright{acmcopyright}
\acmConference[KDD '19]{The 25th ACM SIGKDD Conference on Knowledge Discovery and Data Mining}{August 4--8, 2019}{Anchorage, AK, USA}
\acmBooktitle{The 25th ACM SIGKDD Conference on Knowledge Discovery and Data Mining (KDD '19), August 4--8, 2019, Anchorage, AK, USA}
\acmPrice{15.00}
\acmDOI{10.1145/3292500.3330964}
\acmISBN{978-1-4503-6201-6/19/08}

\newtheorem{definition}{Definition}
\newtheorem{problem}{Problem}

\begin{document}

\settopmatter{printacmref=true}
\fancyhead{}

\title{Representation Learning for Attributed Multiplex Heterogeneous Network}

\author[Y. Cen, X. Zou, J. Zhang, H. Yang, J. Zhou and J. Tang]{
    Yukuo Cen$^{\dagger}$, Xu Zou$^{\dagger}$, Jianwei Zhang$^{\ddagger}$, Hongxia Yang$^{\ddagger*}$, Jingren Zhou$^{\ddagger}$, Jie Tang$^{\dagger*}$
}
\affiliation{
    $^\dagger$ Department of Computer Science and Technology, Tsinghua University
}
\affiliation{
    $^\ddagger$ DAMO Academy, Alibaba Group
}
\email{
  {cyk18, zoux18}@mails.tsinghua.edu.cn
}
\email{
  {zhangjianwei.zjw, yang.yhx, jingren.zhou}@alibaba-inc.com
}
\email{
   jietang@tsinghua.edu.cn
}

\renewcommand{\authors}{Yukuo Cen, Xu Zou, Jianwei Zhang, Hongxia Yang, Jingren Zhou and Jie Tang}
\renewcommand{\shortauthors}{Y. Cen et al.}
\begin{abstract}
\renewcommand{\thefootnote}{\fnsymbol{footnote}}
\footnotetext[1]{Hongxia Yang and Jie Tang are the corresponding authors.}
\renewcommand{\thefootnote}{\arabic{footnote}}

Network embedding (or graph embedding) has been widely used in many real-world applications. However, existing methods mainly focus on networks with single-typed nodes/edges and cannot scale well to handle large networks. Many real-world networks consist of billions of nodes and edges of multiple types, and each node is associated with different attributes. In this paper, we formalize the problem of embedding learning for the \textit{Attributed Multiplex Heterogeneous Network} and propose a unified framework to address this problem. The framework supports both \textit{transductive} and \textit{inductive} learning. We also give the theoretical analysis of the proposed framework, showing its connection with previous works and proving its better expressiveness.
We conduct systematical evaluations for the proposed framework on four different genres of challenging datasets: Amazon, YouTube, Twitter, and \company\footnote{Code is available at \url{https://github.com/cenyk1230/GATNE}.}.
Experimental results demonstrate that with the learned embeddings from the proposed framework, we can achieve statistically significant improvements (e.g., 5.99-28.23\% lift by F1 scores; $p\ll 0.01$, $t-$test) over previous state-of-the-art methods for link prediction.
The framework has also been successfully deployed on the recommendation system of a worldwide leading e-commerce company, \company\ Group. Results of the offline A/B tests on product recommendation further confirm the effectiveness and efficiency of the framework in practice.
\end{abstract}

%
%



\maketitle

\section{Introduction}\label{sec:intro}

\begin{figure*}[t]
	\centering
	\includegraphics[width=0.95\textwidth]{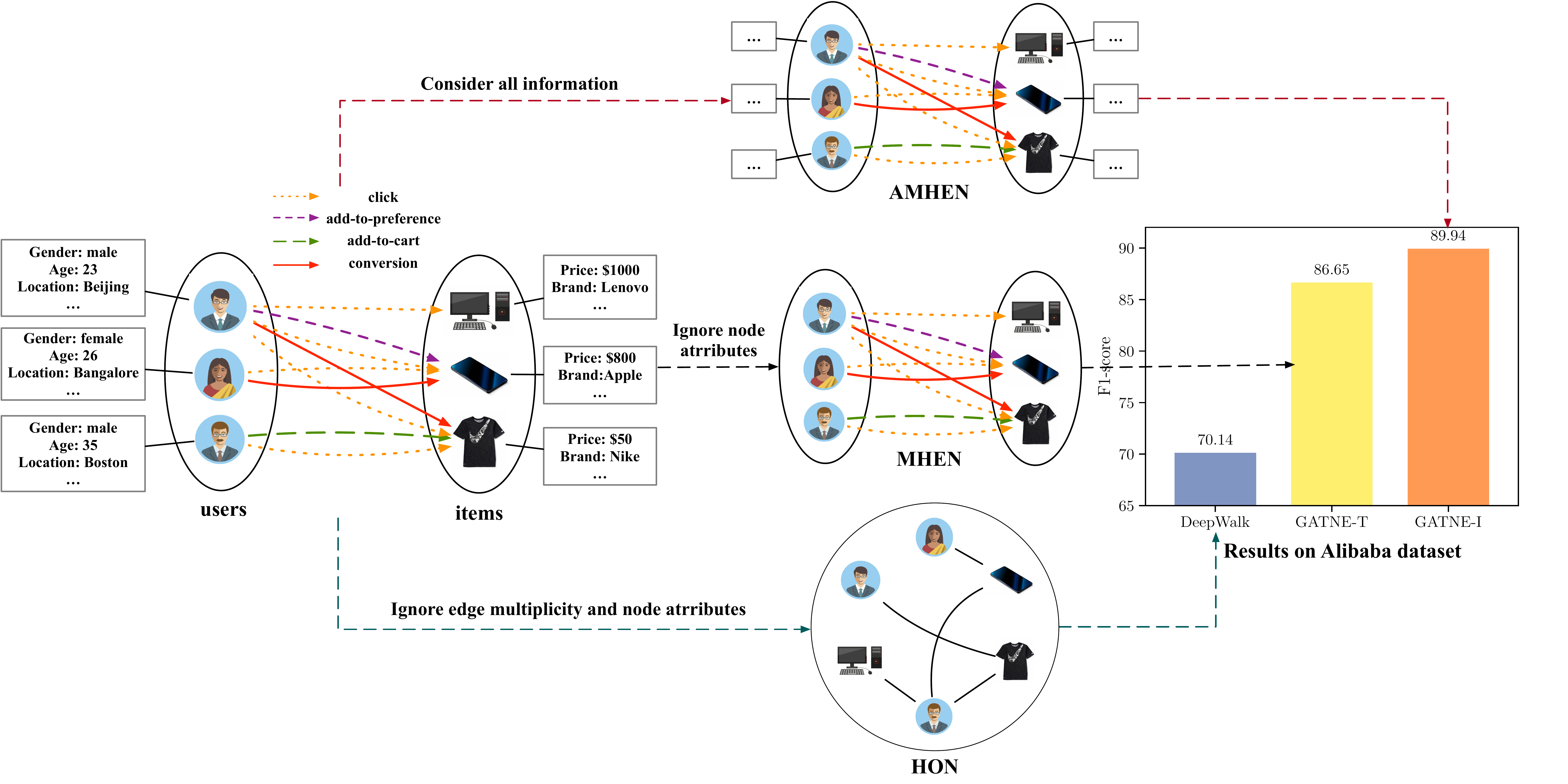}
    \caption{The left illustrates an example of an attributed multiplex heterogeneous network. Users in the left part of the figure are associated with attributes, including gender, age, and location. Similarly, items in the left part of the figure include attributes such as price and brand. The edge types between users and items are from four interactions, including \textit{click}, \textit{add-to-preference}, \textit{add-to-cart} and \textit{conversion}. The three subfigures in the middle represent three different ways of setting up the graphs, including HON, MHEN, and AMHEN from the bottom to the top. The right part shows the performance improvement of the proposed models over DeepWalk
    on the \company\ dataset. As can be seen, \model-I achieves a +28.23\%  performance lift compared to DeepWalk.
} 
	\label{fig:user_item}
\end{figure*}

Network embedding \cite{cui2018survey}, or network representation learning, is a promising method to project nodes in a network to a low-dimensional continuous space while preserving network structure and inherent properties. 
It has attracted tremendous attention recently due to significant progress in downstream network learning tasks such as node classification~\cite{bhagat2011node}, link prediction~\cite{taskar2004link}, and community detection~\cite{fortunato2010community}. 
DeepWalk~\cite{perozzi2014deepwalk}, LINE~\cite{tang2015line}, and node2vec~\cite{grover2016node2vec} are pioneering works that introduce deep learning techniques into network analysis to learn 
node embeddings. NetMF~\cite{qiu2018network}  gives a theoretical analysis of equivalence for the different network embedding algorithms, and later NetSMF~\cite{qiu2019netsmf} gives a scalable solution via sparsification.
Nevertheless, they were designed to handle only the homogeneous network  with single-typed nodes and edges.
More recently, PTE~\cite{tang2015pte}, metapath2vec~\cite{dong2017metapath2vec}, and HERec~\cite{shi2018heterogeneous} are proposed  for heterogeneous networks.
However, real-world  network-structured applications, such as e-commerce, are much more complicated, comprising not only multi-typed nodes and/or edges but also a rich set of  attributes.
Due to its significant importance and challenging requirements, 
there have been tremendous attempts in the literature to investigate embedding learning for  complex networks. Depending on the network topology (homogeneous or heterogeneous) and attributed property (with or without attributes), we categorize six different types of networks and summarize their relative comprehensive developments, respectively, in Table~\ref{tab:methods}. These six categories include \textit{HO}mogeneous \textit{N}etwork (or HON),  \textit{A}ttributed \textit{HO}mogeneous \textit{N}etwork (or AHON),  \textit{HE}terogeneous \textit{N}etwork (or HEN), \textit{A}ttributed \textit{HE}terogeneous \textit{N}etwork (or AHEN), \textit{M}ultiplex \textit{HE}terogeneous \textit{N}etwork (or MHEN), and \textit{A}ttributed \textit{M}ultiplex \textit{HE}terogeneous \textit{N}etwork (or AMHEN).  
As can be seen, among them the AMHEN has been least studied.

In this paper, we focus on embedding learning for AMHENs, where different types of nodes might be linked with multiple different types of edges, and each node is associated with a set of different attributes. This is common in many online applications. For example, in the four datasets that we are working with, there are  20.3\% (Twitter), 21.6\% (YouTube), 15.1\% (Amazon) and 16.3\% (\company) of the linked node pairs having more than one type of edges respectively. As an instance, in an e-commerce system, 
users may have several types of interactions with items, such as click, conversion, add-to-cart, add-to-preference. Figure~\ref{fig:user_item} illustrates such an example.
Obviously, ``users'' and ``items'' have intrinsically different properties and shall not be treated equally.
Moreover, different user-item interactions imply different levels of interests and should be treated differently. Otherwise, the system cannot precisely capture the user's behavioral patterns and preferences and would be insufficient for practical use. 

Not merely because of the heterogeneity and multiplicity, in practice, dealing with  AMHEN poses several unique challenges:

\begin{itemize}[leftmargin=*]
	
	\item  \textit{Multiplex Edges}. Each node pair may have multiple different types of relationships.  It is important to be able to borrow strengths from different relationships and learn unified embeddings.
	
	\item  \textit{Partial Observations}. The real networked data are usually partially observed. For example, a long-tailed customer may only present few interactions with some products. Most existing network embedding methods  focus on the transductive settings, and cannot handle the long-tailed or cold-start problems.
	
	
	\item \textit{Scalability}. Real networks usually have billions of nodes and tens or hundreds of billions of edges ~\cite{wang2018billion}. It is important to develop learning algorithms that can scale well to large networks.
	
\end{itemize}
\begin{table}[t]
	\footnotesize
	\centering
	\caption{\label{tab:methods} The network types handled by different methods.}
	\begin{tabular}{@{\;}c@{\;\;}|@{\;}c@{\;}|@{\;}c@{\;}|@{\;}c@{\;}|c@{\;}}
		\hline \hline
		\multirow{2}*{Network Type} & \multirow{2}*{Method} &  \multicolumn{2}{c|}{Heterogeneity}  & \multirow{2}*{Attribute} \\
		\cline{3-4}
		&&Node Type & Edge Type&\\
		\hline
		\multirow{5}*{\makecell{Homogeneous  \\ Network (HON)}} & DeepWalk~\cite{perozzi2014deepwalk} & \multirow{5}*{Single} & \multirow{5}*{Single} & \multirow{5}*{/}\\
		& LINE~\cite{tang2015line} & & & \\
		& node2vec~\cite{grover2016node2vec} &  & & \\
		& NetMF~\cite{qiu2018network} &  & & \\
		& NetSMF~\cite{qiu2019netsmf} &  & & \\
		\hline
		\multirow{6}*{\makecell{Attributed \\ 
				Homogeneous \\ Network (AHON)}} & TADW~\cite{yang2015network} &  \multirow{6}*{Single} & \multirow{6}*{Single} & \multirow{6}*{Attributed} \\
		& LANE~\cite{huang2017label}  & & & \\
		& AANE~\cite{huang2017accelerated}& & &\\
		& SNE~\cite{liao2018attributed} & & & \\
		& DANE~\cite{gao2018deep}  & & & \\
		& ANRL~\cite{zhang2018anrl}  & & & \\
		\hline
		\multirow{3}*{\makecell{Heterogeneous \\ Network (HEN)}} & PTE~\cite{tang2015pte} &  \multirow{3}*{Multi} & \multirow{3}*{Single} & \multirow{3}*{/} \\
		& metapath2vec~\cite{dong2017metapath2vec}  & & & \\
		& HERec~\cite{shi2018heterogeneous}  & & & \\
		\hline
		\multirow{2}*{\makecell{Attributed \\ HEN (AHEN)}}
		& \multirow{2}*{HNE~\cite{chang2015heterogeneous}} &\multirow{2}*{Multi} &\multirow{2}*{Single} & \multirow{2}*{Attributed} \\
		&&&&\\
		\hline
		\multirow{5}*{\makecell{Multiplex \\ Heterogeneous \\ Network (MHEN)}} & PMNE~\cite{liu2017principled} & \multirow{4}*{Single} & \multirow{4}*{Multi} & \multirow{5}*{/} \\
		& MVE~\cite{qu2017attention}  & & & \\
		& MNE~\cite{ijcai2018-428}  & & & \\
		& mvn2vec~\cite{shi2018mvn2vec} & & & \\
		\cline{2-4}
		& \model-T & Multi & Multi &\\
		\hline
		\multirow{2}*{\makecell{Attributed \\ MHEN (AMHEN)}}
		&\multirow{2}*{\model-I} &  \multirow{2}*{Multi} & \multirow{2}*{Multi} & \multirow{2}*{Attributed} \\
		&  &  &  &  \\
		\hline \hline
	\end{tabular}
\end{table}

To address the above challenges, we propose a novel approach to capture both rich attributed information and to utilize multiplex topological structures from different node types, namely  \textbf{G}eneral \textbf{A}ttributed Multiplex  He\textbf{T}erogeneous \textbf{N}etwork \textbf{E}mbedding, or abbreviated as \model. 
The key features of \model\ are the following:

\begin{itemize}[leftmargin=*]
\item We formally define the problem of attributed multiplex heterogeneous network embedding, which is a more general representation for real-world networks.

\item \model\ supports both transductive and inductive embeddings learning for
attributed multiplex heterogeneous networks. We also give the theoretical analysis to prove that our transductive model is a more general form than existing models (e.g., MNE~\cite{ijcai2018-428}).

\item Efficient and scalable learning algorithms for \model\ have been developed. Our learning algorithms are able to handle hundreds of million nodes and billions of edges efficiently.

\end{itemize}

We conduct extensive experiments to evaluate the proposed models on four different genres of datasets: Amazon, YouTube, Twitter, and \company.
Experimental results show that the proposed framework can achieve statistically significant improvements ($\sim$5.99-28.23\% lift by F1 scores on \company\ dataset; $p \ll$0.01, $t-$test) over state-of-the-art methods. 
We have deployed the proposed model on \company's distributed system and apply the method to \company's recommendation engine. Offline A/B tests further confirm the effectiveness and efficiency of our proposed models.

\hide{
\vpara{Organization.} The rest of this paper is organized as follows. In Section~\ref{sec:related}, we briefly review the related work. In Section~\ref{sec:problem}, we formally define the problem of attributed multiplex heterogeneous network embedding. In Section~\ref{sec:model}, we introduce \model\ in details. In Section~\ref{sec:exp}, we present the experimental results. Finally, we conclude in Section~\ref{sec:conclusion}.
}

\section{Related Work} \label{sec:related}






%
In this section, we review related state-of-the-arts for network embedding, heterogeneous network embedding, multiplex heterogeneous network embedding, and attributed network embedding.


\vpara{Network Embedding.}
Works in network embedding mainly consist of two categories, graph embedding (GE) and graph neural network (GNN). Representative works for GE include 
DeepWalk~\cite{perozzi2014deepwalk} which generates a corpus on graphs by random walk and then trains a skip-gram model on the corpus. LINE~\cite{tang2015line} learns node presentations on large-scale networks while preserving both first-order and second-order proximities. node2vec~\cite{grover2016node2vec} designs a biased random walk procedure to efficiently explore diverse neighborhoods. NetMF~\cite{qiu2018network} is a unified matrix factorization framework for theoretically understanding and improving DeepWalk and LINE.
For popular works in GNN,  GCN~\cite{kipf2017semi} incorporates neighbors' feature representations 
into the node feature representation using convolutional operations.
GraphSAGE~\cite{hamilton2017inductive} provides an inductive approach to combine structural information with node features. It learns functional representations instead of direct embeddings for each node, which helps it work inductively on unobserved nodes during training.

\vpara{Heterogeneous Network Embedding.}
Heterogeneous networks examine scenarios with nodes and/or edges of various types. Such networks are notoriously difficult to mine because of the bewildering combination of heterogeneous contents and structures. The creation of a multidimensional embedding of such data opens the door to the use of a wide variety of off-the-shelf mining techniques for multidimensional data. Despite the importance of this problem, limited efforts have been made on embedding a scalable network of dynamic and heterogeneous data. HNE~\cite{chang2015heterogeneous} jointly considers the contents and topological structures in networks and represents different objects
in heterogeneous networks to unified vector representations. PTE~\cite{tang2015pte} constructs large-scale heterogeneous text network from  labeled information and different levels of word co-occurrence information, which is then embedded into a low-dimensional space. 
metapath2vec~\cite{dong2017metapath2vec} formalizes meta-path based random walk to construct the heterogeneous neighborhood of a node and then leverages a heterogeneous skip-gram model to perform node embeddings. 
HERec~\cite{shi2018heterogeneous} uses a meta-path based random walk strategy to generate meaningful node sequences to learn network embeddings that are first transformed by a set of fusion functions and subsequently integrated into an extended matrix factorization (MF) model.

\vpara{Multiplex Heterogeneous Network Embedding.}
Existing approaches usually study networks with a single type of proximity between nodes, which only captures a single view of a network. However, in reality, there usually exist multiple types of proximities between nodes, yielding networks with multiple views, or multiplex network embedding. PMNE~\cite{liu2017principled} proposes three 
methods to project a multiplex network into a continuous vector space. MVE~\cite{qu2017attention} embeds networks with multiple views in a single collaborated embedding using attention mechanism. 
MNE~\cite{ijcai2018-428} uses one common embedding and several additional embeddings of each edge type for each node, which are jointly learned by a unified network embedding model.
Mvn2vec~\cite{shi2018mvn2vec} explores the feasibility to achieve better embedding results by simultaneously modeling preservation and collaboration to represent semantic meanings of edges in different views respectively.

\vpara{Attributed Network Embedding.}
Attributed network embedding aims to seek for low-dimensional vector representations for nodes in a network, such that original network topological structure and node attribute proximity can be preserved in such representations. 
TADW~\cite{yang2015network} incorporates text features of vertices into network representation learning under the framework of matrix
factorization. 
LANE~\cite{huang2017label} smoothly incorporates label information into the attributed network embedding while preserving their correlations.
AANE~\cite{huang2017accelerated} enables a joint learning process to be done in a distributed manner for accelerated attributed network embedding.
SNE~\cite{liao2018attributed} proposes a generic framework for embedding social networks by capturing both the structural proximity and attribute proximity. 
DANE~\cite{gao2018deep} can capture the high nonlinearity and preserve various proximities in both topological structure and node attributes. ANRL~\cite{zhang2018anrl} uses a neighbor enhancement autoencoder to model the node attribute information and an attribute-aware skip-gram model based on the attribute encoder to capture the network structure.

\section{Problem Definition}\label{sec:problem}
\begin{table}
  \centering
  \caption{\label{tab:notations} Notations.\hide{\green{notation table need refine}}}
  \small
  \begin{tabular}{c|p{2.55in}}
    \hline \hline
    \textbf{Notation} & \textbf{Description} \\
    \hline
    $G$ & the input network\\
    $\mathcal{V}, \mathcal{E}$ & the node/edge set of $G$\\
    $\mathcal{O}, \mathcal{R}$ & the node/edge type set of $G$\\
    $\mathcal{A}$ & the attribute set of $G$\\
    $n$ & the number of nodes\\
    $m$ & the number of edge types\\
    $r$ & an edge type \\
    $d$ & the dimension of base/overall embeddings \\
    $s$ & the dimension of edge embeddings \\
    $v$ & a node in the graph\\
    $\mathcal{N}$ & the neighborhood set of a node on an edge type \\
    $\mathbf{b},\mathbf{u},\mathbf{c},\mathbf{v}$ & the base/edge/context/overall embedding of a node \\
    $\mathbf{h}, \mathbf{g}$ & transformation functions in our inductive approach \\
    $\mathbf{x}$ & the attribute of a node\\
    \hline \hline
  \end{tabular}
\end{table}

Denote a network $G=(\mathcal{V},\mathcal{E})$, where $\mathcal{V}$ is a set of $n$ nodes and $\mathcal{E}$ is a set of edges between nodes. Each edge $e_{ij}=(v_i, v_j)$ is associated with a weight $w_{ij}\ge 0$, indicating the strength of the relationship between $v_i$ and $v_j$. In practice, the network could be either directed or undirected. If $G$ is directed, we have $e_{ij} \nequiv e_{ji}$ and $w_{ij} \nequiv w_{ji}$; if $G$ is undirected, we have $e_{ij}\equiv e_{ji}$ and $w_{ij}\equiv w_{ji}$.  Notations are summarized in Table~\ref{tab:notations}.

\begin{definition}[Heterogeneous Network]
A \textbf{heterogeneous network} \cite{sun2013pathselclus,chang2015heterogeneous} is a network $G=(\mathcal{V},\mathcal{E})$ associated with a node type mapping function $\phi: \mathcal{V} \to \mathcal{O}$ and an edge type mapping function $\psi: \mathcal{E} \to \mathcal{R}$, where $\mathcal{O}$ and $\mathcal{R}$ represent the set of all node types and the set of all edge types, respectively. Each node $v\in \mathcal{V}$ belongs to a particular node type. Similarly, each edge $e\in \mathcal{E}$ is categorized into a specific edge type. 
If $|\mathcal{O}|+|\mathcal{R}|>2$, the network is called \textbf{heterogeneous}; otherwise
\textbf{homogeneous}. 
\end{definition}

Notice that in a heterogeneous network, an edge $e$ can no longer be denoted as $e_{ij}$ since there may be multiple types of edges between node $v_i$ and $v_j$. Under such situations, an edge is denoted as $e_{ij}^{(r)}$, where $r$ corresponds to a certain edge type.

\begin{definition}[Attributed Network]
\label{def:ahn}
An \textbf{attributed network} \cite{huang2017label,chang2015heterogeneous}  is a network $G$ endowed with an attribute representation $\mathcal{A}$, i.e., $G=(\mathcal{V},\mathcal{E},\mathcal{A})$. Each node $v_i\in\mathcal{V}$ is associated with some types of feature vectors. $\mathcal{A}=\{\mathbf{x}_{i}|v_i\in \mathcal{V}\}$ is the set of node features for all nodes, where $\mathbf{x}_i$ is the associated node feature of node $v_i$. 
\end{definition}

\begin{definition}[Attributed Multiplex Heterogeneous Network]
An \textbf{attributed multiplex heterogeneous network} is a network $G=(\mathcal{V},\mathcal{E} ,\mathcal{A})$, $\mathcal{E}=\bigcup_{r\in \mathcal{R}}\mathcal{E}_r$, where $\mathcal{E}_r$ consists of all edges with edge type $r \in \mathcal{R}$, and $|\mathcal{R}|>1$. We separate the network for every edge type or view $r\in\mathcal{R}$ as $G_r=(\mathcal{V},\mathcal{E}_r ,\mathcal{A})$. 
\end{definition}

An example of AMHEN is illustrated in Figure~\ref{fig:user_item}, which contains $2$ node types and $4$ edge types. The two node types are $user$ and $item$ with different attributes. 
Given the above definitions, we can formally define our problem for representation learning on networks.

\begin{problem}[AMHEN Embedding]
Given an AMHEN $G=(\mathcal{V},\mathcal{E},\mathcal{A})$, the problem of \textbf{AMHEN embedding} is to give a unified low-dimensional space representation of each node $v$ on every edge type $r$. The goal is to find a function $f_r: \mathcal{V}\to \mathbb{R}^d$ for every edge type $r$, where $d\ll |\mathcal{V}|$.
\end{problem}

\section{Methodology}\label{sec:model}

In this section, we first explain the proposed \model\ framework in the transductive context~\cite{kipf2017semi}. The resultant model is named as \model-T.  We also give a theoretical discussion about the connection of  \model-T with the newly trending models, e.g., MNE. To deal with the partial observation problem, we further extend the model to the inductive context~\cite{yang2016revisiting} and present a new inductive model named as \model-I. For both models, we present efficient optimization algorithms.



\begin{figure}[t]
    \centering
    \includegraphics[width=0.48\textwidth]{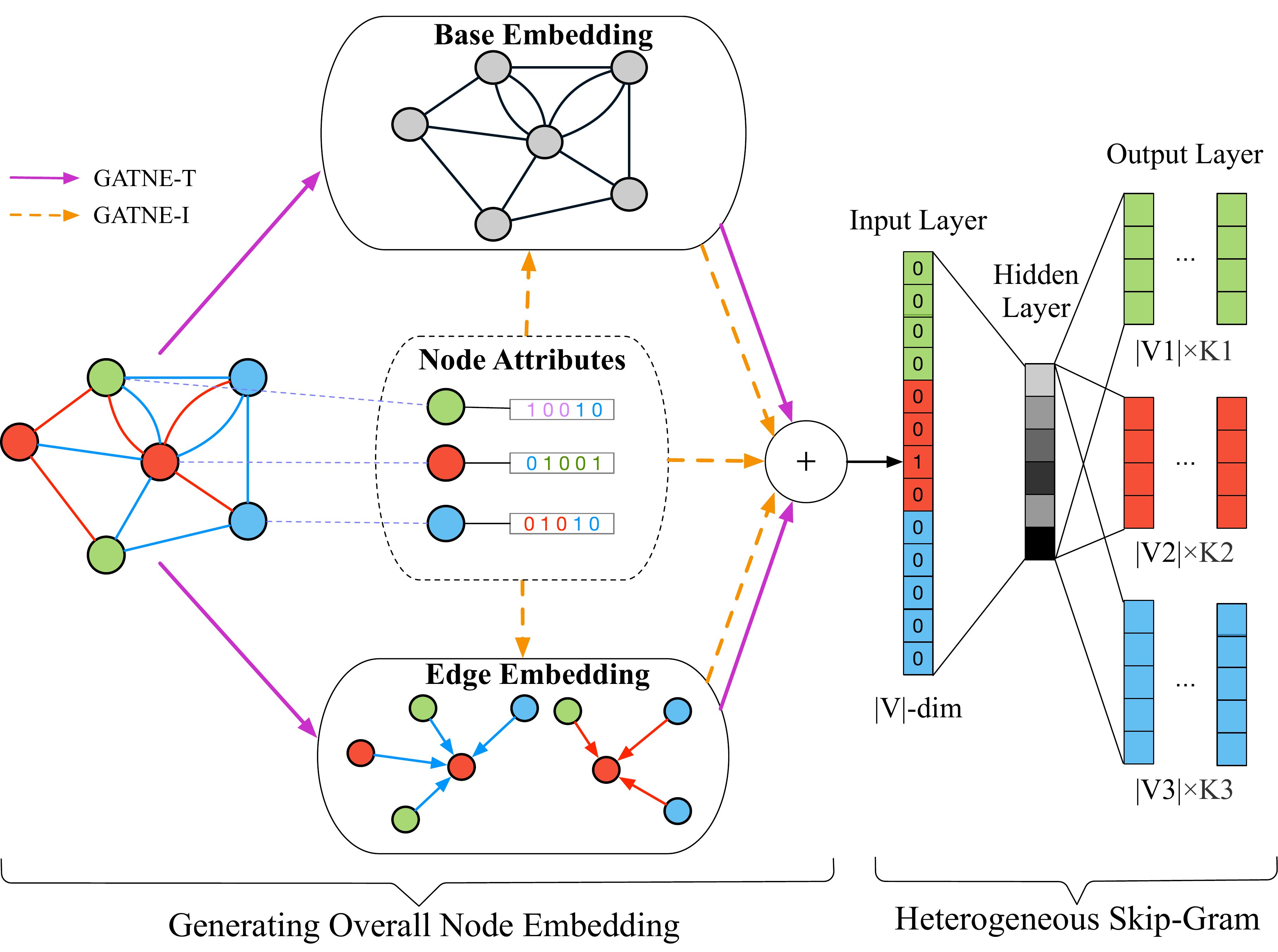}
    \caption{Illustration of \model-T and \model-I models. \model-T only uses network structure information while \model-I considers both structure information and node attributes. The output layer of heterogeneous skip-gram specifies one set of multinomial distributions for each node type in the neighborhood of the input node $v$. In this example, $\mathcal{V}=\mathcal{V}1\cup \mathcal{V}2\cup \mathcal{V}3$ and $K1, K2, K3$ specify the size of $v$'s neighborhood on each node type respectively.}
    \label{fig:model}
\end{figure}

\subsection{Transductive Model: \model-T}
\label{sec:trans}

We begin with embedding learning for 
attributed multiplex heterogeneous networks in the transductive context, and present our model \model-T.
More specifically, in \model-T, 
we split the overall embedding of a certain node $v_i$ on each edge type $r$ into two parts: base embedding, edge embedding as shown in Figure~\ref{fig:model}. The base embedding of node $v_i$ is shared between different edge types. \hide{The edge embedding of node $v_i$ is computed by $m$ aggregated embeddings from neighbors for each edge type, where $m$ is the number of edge types of the network. Each edge embedding $\mathbf{u}_{i,r}\in \mathbb{R}^s$ on edge type $r$ is computed by the meta embeddings of $v_i$'s neighbors corresponding to edge type $r$ as: }
The $k$-th level edge embedding $\mathbf{u}^{(k)}_{i,r} \in \mathbb{R}^s, (1\leq k\leq K)$ of node $v_i$ on edge type $r$ is aggregated from neighbors' edge embeddings as:
\begin{equation}
\label{aggregation}
  \mathbf{u}^{(k)}_{i,r} = aggregator(\{\mathbf{u}^{(k-1)}_{j,r}, \forall v_j\in \mathcal{N}_{i,r}\}),  
\end{equation}

\noindent where \hide{$s$ is the dimension of edge embeddings, $K$ is the total aggregation levels and} $\mathcal{N}_{i,r}$ is the neighbors of node $v_i$ on edge type $r$. The initial edge embedding $\mathbf{u}^{(0)}_{i,r}$ for each node $v_i$ and each edge type $r$ is randomly initialized in our transductive model. Following GraphSAGE~\cite{hamilton2017inductive}, the $aggregator$ function can be a mean aggregator as:
\begin{equation}
\label{mean}
  \mathbf{u}^{(k)}_{i,r} = \sigma(\mathbf{\hat{W}}^{(k)} \cdot \operatorname{mean}(\{\mathbf{u}^{(k-1)}_{j,r}, \forall v_j\in\mathcal{N}_{i,r}\})),
\end{equation}

\noindent or other pooling aggregators such as max-pooling aggregator:
\begin{equation}
\label{pool}
  \mathbf{u}^{(k)}_{i,r} = \max(\{\sigma(\mathbf{\hat{W}}^{(k)}_{pool}\mathbf{u}^{(k-1)}_{j,r}+\mathbf{\hat{b}}^{(k)}_{pool}), \forall v_j\in\mathcal{N}_{i,r}\}),
\end{equation}

\noindent where $\sigma$ is an activation function.
We denote the $K$-th level edge embedding $\mathbf{u}^{(K)}_{i,r}$ as edge embedding $\mathbf{u}_{i,r}$, and concatenate all the edge embeddings for node $v_i$ as $\mathbf{U}_i$ with size $s$-by-$m$, where $s$ is the dimension of edge-embeddings:
\begin{equation}
\mathbf{U}_i = ( \mathbf{u}_{i,1}, \mathbf{u}_{i,2}, \ldots , \mathbf{u}_{i,m} ).
\end{equation}

We use self-attention mechanism~\cite{lin2017structured} to compute the coefficients $\mathbf{a}_{i,r} \in \mathbb{R}^m$ of linear combination of vectors in $\mathbf{U}_i$ on edge type $r$ as:
\begin{equation}
\label{eqn:air}
\mathbf{a}_{i,r} = \operatorname{softmax} ( \mathbf{w}_r^{T} \tanh ~( \mathbf{W}_r \mathbf{U}_i ) )^T,
\end{equation}

\noindent where $\mathbf{w}_r$ and $\mathbf{W}_r$ are trainable parameters for edge type $r$ with size $d_a$ and $d_a\times s$ respectively and the superscript $T$ denotes the transposition of the vector or the matrix. Thus, the overall embedding of node $v_i$ for edge type $r$ is:
\begin{equation}
\label{GATNE}
\mathbf{v}_{i,r} = \mathbf{b}_i + \alpha_r  \mathbf{M}_r^{T} \mathbf{U}_i \mathbf{a}_{i,r},
\end{equation}

\noindent where $\mathbf{b}_i$ is the \bi\ embedding for node $v_i$, $\alpha_r$ is a hyper-parameter denoting the importance of edge embeddings towards the overall embedding and $\mathbf{M}_r \in \mathbb{R}^{s\times d}$ is a trainable transformation matrix.

\vpara{Connection with Previous Work.} 
We choose MNE~\cite{ijcai2018-428}, a recent representative work for MHEN, as the base model for multiplex heterogeneous networks to discuss the connection between our proposed model and previous work. 
In \model-T, we use the attention mechanism to capture the influential factors between different edge types. 
We  theoretically prove that 
our transductive model  
is a more general form of 
MNE and improves the expressiveness.
For MNE, the overall embedding $\mathbf{\tilde{v}}_{i,r}$ of node $v_i$ on edge type $r\in \mathcal{R}$ is 
\begin{equation}
\label{MNE}
\mathbf{\tilde{v}}_{i,r}=\mathbf{b}_i+\alpha_r\mathbf{X}_r^{T}\mathbf{o}_{i,r},
\end{equation}
\noindent where $\mathbf{X}_r$ is a edge-specific transformation matrix.
And for \model-T, the overall embedding for node $v_i$ on edge type $r$ is:
\begin{equation}
\begin{aligned}
\mathbf{v}_{i,r} &= \mathbf{b}_i+\alpha_r  \mathbf{M}_r^{T} \mathbf{U}_i \mathbf{a}_{i,r} = \mathbf{b}_i + \alpha_r \mathbf{M}_r^{T} \sum_{p=1}^{m}\lambda_{p} \mathbf{u}_{i,p},\\
\end{aligned}
\end{equation}

\noindent where $\lambda_p$ denotes the $p$-th element of $\mathbf{a}_{i,r}$ and is computed as:
\begin{equation}
\label{lamb}
	\lambda_{p} = \frac{\exp (\mathbf{w}_r^{T} \tanh ~( \mathbf{W}_r \mathbf{u}_{i,p}))}{\sum_t\exp (\mathbf{w}_r^{T} \tanh ~( \mathbf{W}_r \mathbf{u}_{i,t}))}. 
\end{equation}

\begin{theorem}
\label{thm}
For any $r\in\mathcal{R}$, there exist parameters $\mathbf{w}_r$ and $\mathbf{W}_r$, such that for any $\mathbf{o}_{i,1}, \ldots, \mathbf{o}_{i,m}$, and corresponding matrix $\mathbf{X}_r$, with dimension $s$ for each $\mathbf{o}_{i,j}$ and size $s\times d$ for $\mathbf{X}_r$, there exist $\mathbf{u}_{i,1}, \ldots, \mathbf{u}_{i,m}$, and corresponding matrix $\mathbf{M}_r$, with dimension $s+m$ for each $\mathbf{u}_{i,j}$ and size $(s+m)\times d$ for $\mathbf{M}_r$, that satisfy $\mathbf{v}_{i,r}\approx \mathbf{\tilde{v}}_{i,r}$.
\end{theorem}

\begin{proof}
For any $t$, let $\mathbf{u}_{i,t}=(\mathbf{o}^T_{i,t},\mathbf{u}'^T_{i,t})^T$ concatenated by two vectors, where the first $s$ dimension is $\mathbf{o}_{i,t}$, and $\mathbf{u}_{i,t}'$ is an $m$-dimensional vector. 
Let $u_{i,t,k}'$ denote the $k$-th dimension of $\mathbf{u}_{i,t}'$, then take $u_{i,t,k}'=\delta_{tk}$, where $\delta$ is the Kronecker delta function as
\begin{equation}
\label{kronecker}
\delta_{ij}=1, i=j; \delta_{ij}=0, i\neq j.
\end{equation}
Let $\mathbf{W}_r$ be all zero, except for the element on row $1$ and column $s+r$ with a large enough number $M$; therefore $\tanh(\mathbf{W}_r \mathbf{u}_{i,p})$ becomes a vector with its $1^{st}$ dimension approximately being $\delta_{rp}$, and other dimensions being $0$. Take $\mathbf{w}_r$ as a vector with its $1^{st}$ dimension being $M$,  then $\exp (\mathbf{w}_r^T \tanh ~( \mathbf{W}_r \mathbf{u}_{i,p}))\approx \exp(M\delta_{rp})$, so $\lambda_{p}\approx \delta_{rp}$.  Finally we take $\mathbf{M}_r$ being $\mathbf{X}_r$ at its first $s\times d$ dimension, and $0$ on the following $m\times d$ dimension, and we can get $\mathbf{v}_{i,r}\approx \mathbf{\tilde{v}}_{i,r}$.
\end{proof}
Thus the parameter space of MNE is almost included by our model's parameter space and we can conclude that \model-T is a generalization of MNE, if edge embeddings can be trained directly. However, in our model, the edge embedding $\mathbf{u}$ is generated from single or multiple layers of aggregation. We give more discussions about the aggregation case. \hide{Interested readers are referred to Appendix \ref{extension} for more discussions about the aggregated case.} 

\vpara{Effects of Aggregation.}
In the \model-T model, the edge embedding $\mathbf{u}^{(k)}$ is computed by aggregating the edge embedding $\mathbf{u}^{(k-1)}$ of its neighbors as:
\begin{equation}
\label{eqn:uirk}
\mathbf{u}^{(k)}_{i,r}=\sigma(\mathbf{\hat{W}}^{(k)}\cdot\operatorname{mean}(\{\mathbf{u}^{(k-1)}_{j,r}, v_j\in\mathcal{N}_{i,r}\})).
\end{equation}
The mean aggregator is basically a matrix multiplication,
\begin{equation}
\operatorname{mean}_{k}(v_i)=\operatorname{mean}(\{\mathbf{u}^{(k-1)}_{j,r}, v_j\in\mathcal{N}_{i,r}\})=(\mathbf{U}^{(k-1)}_{r}\mathbf{N}_r)_i,
\end{equation}
\noindent where $\mathbf{N}_r$ is the neighborhood matrix on edge type $r$, $\mathbf{U}^{(k)}_r=(\mathbf{u}^{(k)}_{1,r},...,\mathbf{u}^{(k)}_{n,r})$ is the $k$-th level edge embedding of all nodes in the graph on edge type $r$, and $(\mathbf{U}^{(k-1)}_{r}\mathbf{N}_r)_i$ denotes the $i^{th}$ column of $\mathbf{U}^{(k-1)}_{r}\mathbf{N}_r$. $\mathbf{N}_r$ can be a normalized adjacency matrix. The mean operator of Equation (\ref{eqn:uirk}) can be weighted and the neighborhood matrix $\mathcal{N}_{i,r}$ can be sampled.
Take $\mathbf{\hat{W}}^{(k)}=\mathbf{I}$, where $\mathbf{I}$ is an identity matrix, then $\mathbf{u}^{(k)}_{i,r}=\sigma(\operatorname{mean}_k(v_i))$.
If $\mathbf{N}_r$ is of full rank, then for any $\mathbf{O}_r=(\mathbf{o}_{1,r}, \ldots, \mathbf{o}_{n,r})$, there exists $\mathbf{U}^{(k-1)}_r$ such that $\mathbf{U}^{(k-1)}_{r}\mathbf{N}_r=\mathbf{O}_r$.

If the activation function $\sigma$ is an automorphism, i.e., $\sigma:\mathbb{R}\rightarrow \mathbb{R}$ and $\mathbf{N}_r$ is of full rank, we can use the construction method described in Theorem \ref{thm} to construct $\mathbf{u}_{i,r}$ and the above method to construct each level of edge embeddings $\mathbf{u}^{(K-1)}_{i,r}, ..., \mathbf{u}^{(0)}_{i,r}$ subsequently. Therefore, our model is still a more general form that can generalize the MNE model, when all the neighborhood matrices $\mathbf{N}_r$ and the activation function $\sigma$ are invertible in all levels of aggregation. 

\subsection{Inductive Model: \model-I}
The limitation of \model-T is that it cannot handle unobserved data. 
However, in many real-world applications, the networked data is often partially observed~\cite{yang2016revisiting}.
We then extend our model to the inductive context and present a new model named \model-I. The model is also illustrated in Figure~\ref{fig:model}. 
Specifically, we define the \bi\ embedding $\mathbf{b}_i$ as a parameterized function of $v_i$'s attribute $\mathbf{x}_i$ as $\mathbf{b}_i = \mathbf{h}_z(\mathbf{x}_i)$, where $\mathbf{h}_z$ is a transformation function\hide{ that transforms $\mathbf{x}_i$ to \bi\ embedding $\mathbf{b}_i$} and $z=\phi(i)$ is node $v_i$'s corresponding node type. Notice that nodes with different types may have different dimensions of their attributes $\mathbf{x}_i$. The transformation function $\mathbf{h}_z$ can have different forms such as a multi-layer perceptron~\cite{pal1992multilayer}.
Similarly, the initial \gij\ embeddings $\mathbf{u}^{(0)}_{i,r}$ for node $i$ on edge type $r$ should be also parameterized as the function of attributes $\mathbf{x}_i$ as $\mathbf{u}^{(0)}_{i,r}=\mathbf{g}_{z,r}(\mathbf{x}_i)$, where $\mathbf{g}_{z,r}$ is also a transformation function that transforms the feature to an edge embedding of node $v_i$ on the edge type $r$ and $z$ is $v_i$'s corresponding node type. To be more specific, for the inductive model, we also add an extra attribute term to the overall embedding of node $v_i$ on type $r$: 
\begin{equation}
\label{eqn:vni}
\mathbf{v}_{i,r} = \mathbf{h}_z(\mathbf{x}_i) + \alpha_r \mathbf{M}_r^{T} \mathbf{U}_i \mathbf{a}_{i,r} + \beta_r \mathbf{D}_{z}^T \mathbf{x}_i,
\end{equation}

\noindent where $\beta_r$ is a coefficient and $\mathbf{D}_{z}$ is a feature transformation matrix on $v_i$'s corresponding node type $z$. \hide{Each node $v_i$ corresponds to a certain node type $z$ and corresponds to a feature transformation matrix $\mathbf{D}_{z}$.} The difference between our transductive and inductive model mainly lies on how the base embedding $\mathbf{b}_i$ and initial edge embeddings $\mathbf{u}^{(0)}_{i,r}$ are generated. In our transductive model, the base embedding $\mathbf{b}_i$ and initial edge embedding $\mathbf{u}^{(0)}_{i,r}$ are directly trained for each node based on the network structure, and the transductive model cannot handle nodes that are not seen during training. As for our inductive model, instead of training $\mathbf{b}_i$ and $\mathbf{u}^{(0)}_{i,r}$ directly for each node, we train transformation functions $\mathbf{h}_z$ and $\mathbf{g}_{z,r}$ that transforms the raw feature $\mathbf{x}_i$ to $\mathbf{b}_i$ and $\mathbf{u}^{(0)}_{i,r}$, which works for nodes that did not appear during training as long as they have corresponding raw features. \hide{For further discussion of the inductive model see Appendix \ref{extension}.}

\begin{algorithm}[t]
	\caption{\model \label{algo:our}}
	\KwIn{network $G=(\mathcal{V},\mathcal{E},\mathcal{A})$, embedding dimension $d$, edge embedding dimension $s$, learning rate $\eta$, negative samples $L$, coefficient $\alpha$, $\beta$. }
	\KwOut{overall embeddings $\mathbf{v}_{i,r}$ for all nodes on every edge type $r$}
	Initialize all the model parameters $\theta$.\hide{ $\mathbf{b}$, $\mathbf{u}^{(0)}$ for all nodes, and other model parameters~($\mathbf{\hat{W}}$, $\mathbf{\hat{b}}$, $\mathbf{D}$, $\mathbf{M}$, $\mathbf{w}$, $\mathbf{W}$)} \\
	Generate random walks on each edge type $r$ as $\mathcal{P}_r$. \\
	Generate training samples $\{(v_i,v_j,r)\}$ from random walks $\mathcal{P}_r$ on each edge type $r$.\\
	\While{not converged} {
		\ForEach{$(v_i,v_j,r)\in$ training samples} {
			Calculate $\mathbf{v}_{i,r}$ using Equation (\ref{GATNE}) or (\ref{eqn:vni}) \\
			Sample $L$ negative samples and calculate objective function $E$ using Equation (\ref{eqn:obj_fun}) \\
			Update model parameters $\theta$ by $\frac{\partial E}{\partial \theta}$.
		}
	}
\end{algorithm}

\subsection{Model Optimization}
We discuss how to learn the proposed transductive and inductive models.
Following~\cite{perozzi2014deepwalk,tang2015line,grover2016node2vec}, we use random walk to generate node sequences and then perform skip-gram~\cite{mikolov2013efficient,mikolov2013distributed} over the node sequences to learn embeddings. Since each view of the input network is heterogeneous, we use meta-path-based random walks~\cite{dong2017metapath2vec}.
To be specific, given a view $r$ of the network, i.e., $G_r=(\mathcal{V}, \mathcal{E}_r, \mathcal{A})$ and a meta-path scheme $\mathcal{T}: \mathcal{V}_1 \to \mathcal{V}_2 \to \cdots \mathcal{V}_t \cdots \to \mathcal{V}_l$, where $l$ is the length of the meta-path scheme, the transition probability at step $t$ is defined as follows:
\begin{equation}
  p(v_j|v_i, \mathcal{T})=\left\{
    \begin{array}{cl}
      \frac{1}{|\mathcal{N}_{i,r}\cap\mathcal{V}_{t+1}|} & (v_i,v_j)\in \mathcal{E}_r, v_j\in\mathcal{V}_{t+1}, \\
      0 & (v_i,v_j)\in \mathcal{E}_r, v_j\notin\mathcal{V}_{t+1}, \\
      0 & (v_i,v_j)\notin \mathcal{E}_r,
    \end{array}
  \right.
\end{equation}

\noindent where $v_i \in \mathcal{V}_t$ and $\mathcal{N}_{i,r}$ denotes the neighborhood of node $v_i$ on edge type $r$. The flow of the walker is conditioned on the pre-defined meta path $\mathcal{T}$. The meta-path-based random walk strategy ensures that the semantic relationships between different types of nodes can be properly incorporated into skip-gram model~\cite{dong2017metapath2vec}. 
Supposing the random walk with length $l$ on edge type $r$ follows a path $P=(v_{p_1},\ldots,v_{p_l})$ such that $(v_{p_{t-1}}, v_{p_t})\in \mathcal{E}_r (t=2\ldots l)$, denote $v_{p_t}$'s context as $C=\{v_{p_k}|v_{p_k}\in P, |k-t|\leq c, t\neq k\}$, where $c$ is the radius of the window size.

Thus, given a node $v_i$ with its context $C$ of a path, our objective is to minimize the following negative log-likelihood:
\begin{equation}
-\log P_{\theta} \left(\{v_{j}| v_{j}\in C\}|v_{i} \right)=\sum_{v_j\in C} - \log P_{\theta}(v_{j}|v_{i}),
\end{equation}

\noindent where $\theta$ denotes all the parameters. Following metapath2vec~\cite{dong2017metapath2vec} we use the heterogeneous softmax function which is normalized with respect to the node type of node $v_j$. Specifically, the probability of $v_j$ given $v_i$ is defined as:
\begin{equation}
P_{\theta}(v_j|v_i) = \frac{\exp (\mathbf{c}_{j}^{T} \cdot \mathbf{v}_{i,r})}{\sum_{k\in \mathcal{V}_t} \exp (\mathbf{c}_{k}^{T} \cdot \mathbf{v}_{i,r})},
\end{equation}

\noindent where $v_j\in \mathcal{V}_t$, $\mathbf{c}_k$ is the context embedding of node $v_k$ and $\mathbf{v}_{i}$ is the overall embedding of node $v_i$ for edge type $r$.



Finally, we use heterogeneous negative sampling to approximate the objective function $-\log P_{\theta}(v_j|v_i)$ for each node pair $(v_i,v_j)$ as:
\begin{equation}
\label{eqn:obj_fun}
E = - \log \sigma ( \mathbf{c}_{j}^{T} \cdot \mathbf{v}_{i,r} ) - \sum_{l=1}^{L} \mathbb{E}_{v_k\thicksim P_t(v)} [ \log \sigma ( - \mathbf{c}_k^{T} \cdot \mathbf{v}_{i,r} ) ],
\end{equation}

\noindent where $\sigma(x)=1/(1+\exp(-x))$ is the sigmoid function, $L$ is the number of negative samples correspond to a positive training sample, and $v_k$ is randomly drawn from a noise distribution $P_t(v)$ defined on node $v_j$'s corresponding node set $\mathcal{V}_t$. 


We summarize our algorithm in Algorithm~\ref{algo:our}. The time complexity of our random walk based algorithm is $O(nmdL)$ where $n$ is the number of nodes, $m$ is the number of edge types, $d$ is overall embedding size, $L$ is the number of negative samples per training sample ($L\geq 1$). The memory complexity of our algorithm is $O(n(d+m\times s))$ with $s$ being the size of edge embedding. 

\section{Experiments}\label{sec:exp}

In this section, we first introduce the details of four evaluation datasets and the competitor algorithms. We focus on the link prediction task to evaluate the performances of our proposed methods compared to other state-of-the-art methods. Parameter sensitivity, convergence, and scalability are then discussed. Finally, we report the results of offline A/B tests of our method on \company's recommendation system.

\begin{table}
  \centering
  \caption{\label{tab:stats} Statistics of Datasets.}
  \begin{tabular}{c|c|c|c|c}
    \hline \hline
    \textbf{Dataset} & \# nodes & \# edges & \ \# n-types & \ \# e-types \\
    \hline
    Amazon & 10,166 & 148,865 & 1 & 2 \\
    YouTube & 2,000 & 1,310,617 & 1 & 5\\
    Twitter & 10,000 & 331,899 & 1 & 4\\
    \company-S & 6,163 & 17,865 & 2 & 4\\ 
    \company & 41,991,048 & 571,892,183 & 2 & 4\\
    \hline \hline
  \end{tabular}
\end{table}

\subsection{Datasets}
We work on three public datasets and the \company\ dataset for the link prediction task. Amazon Product Dataset\footnote{http://jmcauley.ucsd.edu/data/amazon/}~\cite{mcauley2015image, he2016ups} includes product metadata and links between products; YouTube dataset\footnote{http://socialcomputing.asu.edu/datasets/YouTube}~\cite{tang2009uncoverning, tang2009uncovering}  consists of various types of interactions;  Twitter dataset\footnote{https://snap.stanford.edu/data/higgs-twitter.html}~\cite{de2013anatomy}  also contains various types of links. \company\ dataset has two node types, user and item (or product), and includes four types of interactions between users and items. Since some of the baselines cannot scale to the whole graph, we evaluate performances on sampled datasets. The statistics of these four sampled datasets are summarized in Table~\ref{tab:stats}. 
Notice that \textit{n-types} and \textit{e-types} in the table denote node types and edge types, respectively. 

\vpara{Amazon.} 
In our experiments, we only use the product metadata of \textit{Electronics} category, including the product attributes and co-viewing, co-purchasing links between products. The product attributes include the price, sales-rank, brand, and category. \hide{The node type set of Amazon is $\mathcal{O}=\{product\}$ and the edge type set of Amazon is $\mathcal{R}=\{also\_bought, also\_viewed\}$, denoting that two products are co-bought or co-viewed by the same user respectively. The number of products in \textit{Electronics} is still large for many baselines; therefore, we extract a connected subgraph from the whole graph. }

\vpara{YouTube.} 
YouTube dataset is a multiplex bidirectional network dataset that consists of five types of interactions between 15,088 YouTube users. The types of edges include contact, shared friends, shared subscription, shared subscriber, and shared favorite videos between users. \hide{The size of node type set and edge type set of YouTube dataset is $|\mathcal{O}|=1$ and $|\mathcal{R}|=5$ respectively.}

\vpara{Twitter.} 
Twitter dataset is about tweets related to the discovery of the Higgs boson between 1st and 7th, July 2012. It is made up of four directional relationships between more than 450,000 Twitter users. The relationships are re-tweet, reply, mention, and friendship/follower relationships between Twitter users. \hide{The size of node type set and edge type set of Twitter dataset is $|\mathcal{O}|=1$ and $|\mathcal{R}|=4$ respectively.}

\vpara{\company.} \company\ dataset consists of four types of interactions including click, add-to-preference, add-to-cart, and conversion between two types of nodes, user and item. 
The sampled \company\ dataset is denoted as \textbf{\company-S}. We also provide the evaluation of the whole dataset on \company's distributed cloud platform; the full dataset is denoted as \textbf{\company}.
\subsection{Competitors}
We categorize our competitors into the following four groups. The overall embedding size is set to 200 for all methods. The specific hyper-parameter settings for different methods are listed in the Appendix. \hide{For random-walk based methods, the number of walks for each node is set to 20 and the length of walks is set to 10. The window size is set to 5 for generating node contexts. The number of negative samples is set to 5. The number of iterations for training skip-gram model is set to 10. The number of epochs is set to 5. The additional embedding size for MNE and the edge-specific embedding size for our model is 10. Model specific hyper-parameter settings for different methods are introduced below. }


\vpara{Network Embedding Methods.}
The compared methods include DeepWalk~\cite{perozzi2014deepwalk}, LINE~\cite{tang2015line}, and node2vec~\cite{grover2016node2vec}. As these methods can only deal with HON, we feed separate graphs with different edge types to them and obtain different node embeddings for each separate graph. 

\vpara{Heterogeneous Network Embedding Methods.}
We focus on the representative work  metapath2vec~\cite{dong2017metapath2vec}, which is designed to deal with the node heterogeneity. When there is only one node type in the network, metapath2vec degrades to DeepWalk. 
For \company\ dataset, the meta-path schemes are set to $U-I-U$ and $I-U-I$, where $U$ and $I$ denote \textit{User} and \textit{Item} respectively.

\vpara{Multiplex Heterogeneous Network Embedding Methods.}
The compared methods include PMNE~\cite{liu2017principled}, MVE~\cite{qu2017attention}, MNE~\cite{ijcai2018-428}. We denote the three methods of PMNE as PMNE(n), PMNE(r) and PMNE(c) respectively. MVE uses collaborated context embeddings and applies an attention mechanism to view-specific embedding. MNE uses one common embedding and several additional embeddings for each edge type, which are jointly learned by a unified network embedding model. \hide{We further examined the aggregated version of MNE, denoted as MNE-aggregate, which is basically MNE with the aggregator in section \ref{sec:trans} being added on top of it.}

\vpara{Attributed Network Embedding Methods.}
The compared method is ANRL~\cite{zhang2018anrl}. ANRL uses a neighbor enhancement auto-encoder to model the node attribute information and an attribute-aware skip-gram model based on the attribute encoder to capture the network structure. 


\vpara{Our Methods.}
Our proposed methods include \model-T and \model-I. \model-T considers the network structure and uses base embeddings and edge embeddings to capture the influential factors between different edge types. \model-I considers both the network structure and the node attributes, and learns an inductive transformation function instead of learning base embeddings and meta embeddings for each node directly. For \company\ dataset, we use the same meta-path schemes as metapath2vec. For some datasets without node attributes, we also generate node features for them.
\hide{To analysis the effect of the aggregation and attention mechanism used in our model, we further conduct experiments on subtracting attention/aggregation part of our model, Namely \model-NoAttn and \model-NoAggr. Start from the \model-T model, for \model-NoAttn, after generating $\mathbf{u}$, we do not use the attention mechanism described in equation \ref{GATNE} and use Equation (\ref{MNE}) to generate overall embedding; for \model-NoAggr, we do not use aggregation as Equation (\ref{aggregation}) but simply take $\mathbf{u}_{i,j}=\mathbf{g}_{i,j}$. }
Due to the size of the \company\  dataset with more than 40 million nodes and 500 million edges and the scalabilities of the other competitors, we only compare our \model\ model with DeepWalk, MVE, and MNE. Specific implementations can be found in the Appendix. 


\begin{table*}
	\centering
	\small
	\setlength{\abovecaptionskip}{10pt}
	\caption{\label{tab:tab_res}Performance comparison of different methods on four datasets. \hide{We use \textbf{bold} to highlight methods that surpass all other compared methods.}}
	\begin{tabular}{c|ccc|ccc|ccc|ccc}
		\hline \hline
		& \multicolumn{3}{c|}{Amazon} & \multicolumn{3}{c|}{YouTube} & \multicolumn{3}{c|}{Twitter} & \multicolumn{3}{c}{\company-S} \\
		& ROC-AUC & PR-AUC & F1 & ROC-AUC & PR-AUC & F1 & ROC-AUC & PR-AUC & F1 & ROC-AUC & PR-AUC & F1 \\
		\hline
		DeepWalk & 94.20 & 94.03 & 87.38 & 71.11 & 70.04 & 65.52 & 69.42 & 72.58 & 62.68 & 59.39 & 60.62 & 56.10 \\
		node2vec & 94.47 & 94.30 & 87.88 & 71.21 & 70.32 & 65.36 & 69.90 & 73.04 & 63.12 & 62.26 & 63.40 & 58.49 \\
		LINE & 81.45 & 74.97 & 76.35 & 64.24 & 63.25 & 62.35 & 62.29 & 60.88 & 58.18 & 53.97 & 54.65 & 52.85 \\
		\hline
		metapath2vec & 94.15 & 94.01 & 87.48 & 70.98 & 70.02 & 65.34 & 69.35 & 72.61 & 62.70 & 60.94 & 61.40 & 58.25\\
		\hline
        ANRL & 71.68 & 70.30 & 67.72 & 75.93 & 73.21 & 70.65 & 70.04 & 67.16 & 64.69 & 58.17 & 55.94 & 56.22 \\
		\hline
		PMNE(n) & 95.59 & 95.48 & 89.37 & 65.06 & 63.59 & 60.85 & 69.48 & 72.66 & 62.88 & 62.23 & 63.35 & 58.74 \\
		PMNE(r) & 88.38 & 88.56 & 79.67 & 70.61 & 69.82 & 65.39 & 62.91 & 67.85 & 56.13 & 55.29 & 57.49 & 53.65 \\
		PMNE(c) & 93.55 & 93.46 & 86.42 & 68.63 & 68.22 & 63.54 & 67.04 & 70.23 & 60.84 & 51.57 & 51.78 & 51.44 \\
		MVE & 92.98 & 93.05 & 87.80 & 70.39 & 70.10 & 65.10 & 72.62 & 73.47 & 67.04 & 60.24 & 60.51 & 57.08 \\
		MNE & 90.28 & 91.74 & 83.25 & 82.30 & 82.18 & 75.03 & 91.37 & 91.65 & 84.32 & 62.79 & 63.82 & 58.74 \\
		\hline
		\model-T & \textbf{97.44} & \textbf{97.05} & \textbf{92.87} & \textbf{84.61} & 81.93 & \textbf{76.83} & \textbf{92.30} & 91.77 & \textbf{84.96} & 66.71 & 67.55 & 62.48 \\
		\model-I & 96.25 & 94.77 & 91.36 & 84.47 & \textbf{82.32} & \textbf{76.83} & 92.04 & \textbf{91.95} & 84.38 & \textbf{70.87} & \textbf{71.65} & \textbf{65.54}\\
		\hline \hline
	\end{tabular}
\end{table*}

\subsection{Performance Analysis}
Link prediction is a common task in both academia and industry. For academia, it is widely used to evaluate the quality of network embeddings obtained by different methods. In the industry, link prediction is a very demanding task since in real-world scenarios we are usually facing graphs with partial links, especially for e-commerce companies that rely on the links between their users and items for recommendations. 
We hide a set of edges/non-edges from the original graph and train on the remaining graph. Following~\cite{kipf2016variational, bojchevski2018deep}, we create a validation/test set that contains 5\%/10\% randomly selected positive edges respectively with the equivalent number of randomly selected negative edges for each edge type. The validation set is used for hyper-parameter tuning and early stopping. The test set is used to evaluate the performance and is only run once under the tuned hyper-parameter. We use some commonly used evaluation criteria, i.e., the area under the ROC curve (ROC-AUC)~\cite{hanley1982meaning} and the PR curve (PR-AUC)~\cite{davis2006relationship} in our experiments. 
We also use F1 score as the other metric for evaluation. To avoid the thresholding effect~\cite{tang2009large}, we assume that the number of hidden edges in the test set is given~\cite{tang2009large, perozzi2014deepwalk, qiu2018network}. All of these metrics are uniformly averaged among the selected edge types. 

\begin{table}[t]
    \centering
    \caption{\label{tab:tab_online_res}The experimental results on \company\ dataset. \hide{We use \textbf{bold} to highlight methods that surpass all other compared methods.}} 
    \begin{tabular}{c|ccc}
        \hline \hline
        & ROC-AUC & PR-AUC & F1 \\
        \hline
        DeepWalk & 65.58 & 78.13 & 70.14\\
        MVE & 66.32 & 80.12 & 72.14\\
        MNE & 79.60 & 93.01 & 84.86\\ 
        \model-T & 81.02 & 93.39 & 86.65 \\
        \model-I & \textbf{84.20} & \textbf{95.04} & \textbf{89.94} \\
        \hline \hline
    \end{tabular}
\end{table}

\begin{figure}
    \setlength{\belowcaptionskip}{-0.1cm}
	\centering
	\subfigure[Convergence]{\label{subfig:convergence}\includegraphics[width=0.23\textwidth]{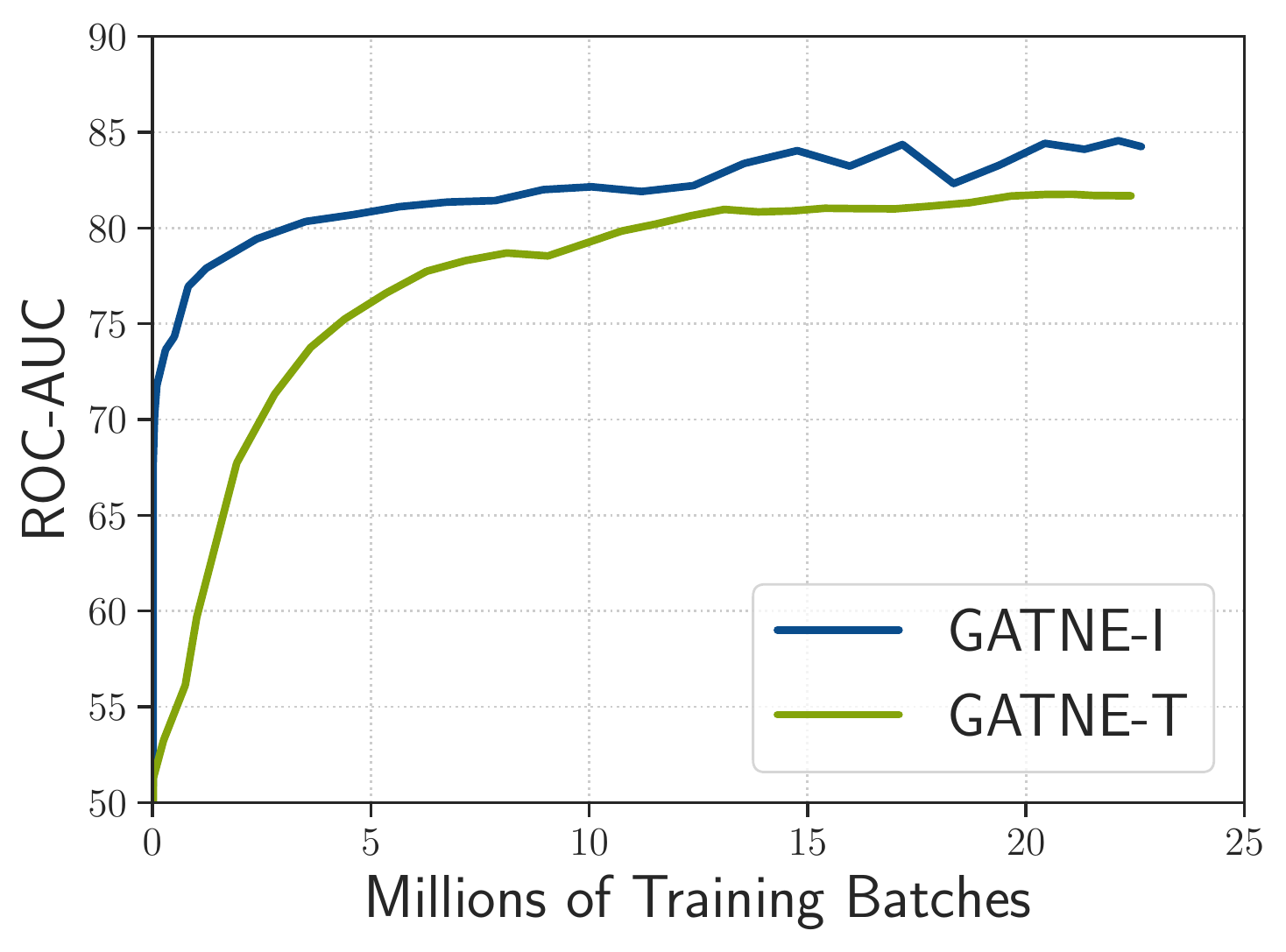}}
 	\subfigure[Scalability]{\label{subfig:scalability}\includegraphics[width=0.23\textwidth]{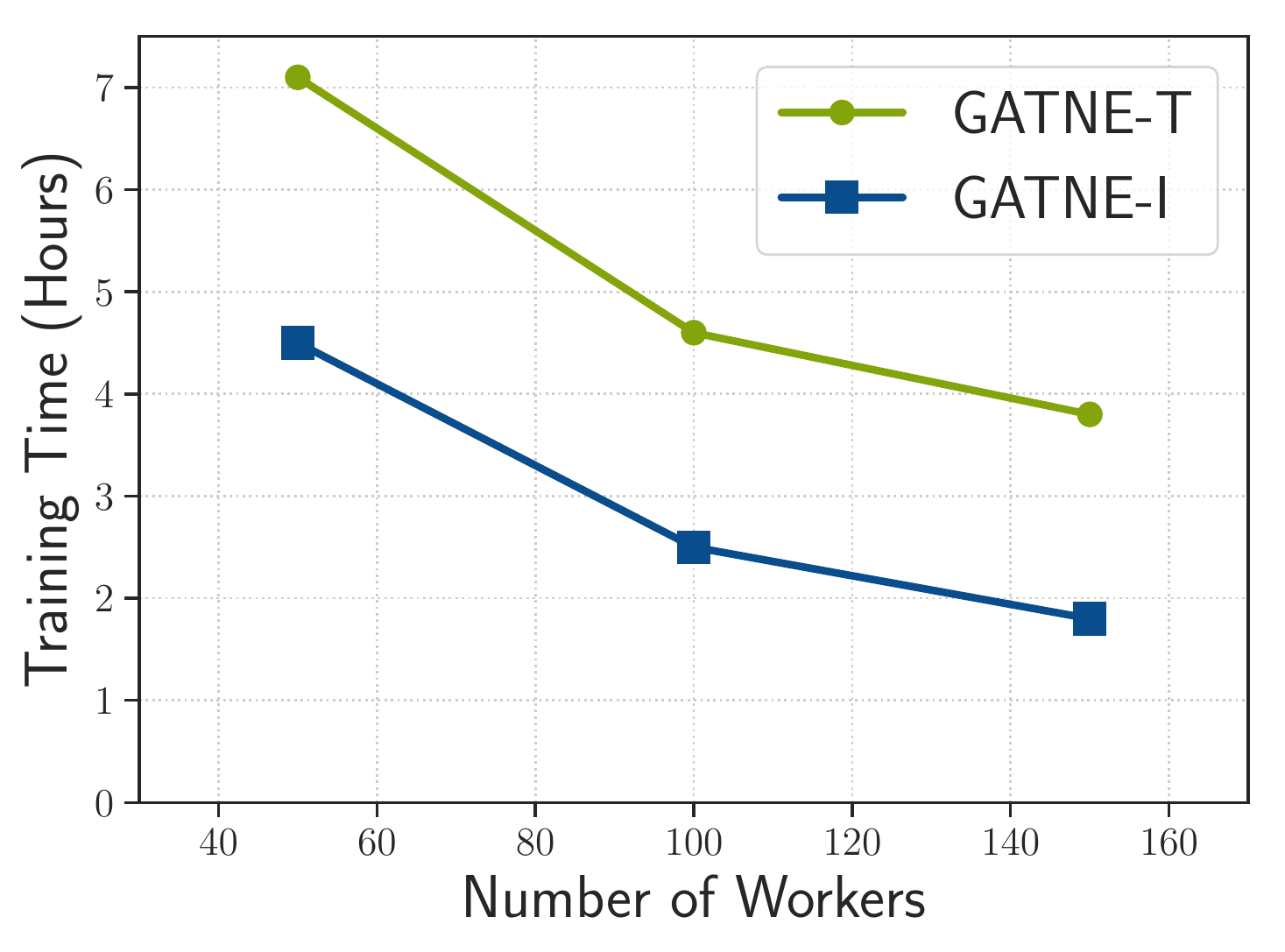}}
	\caption{(a) The convergence curve for  \model-T and \model-I models on \company\ dataset. The inductive model converges faster and achieves better performance than the transductive model. \hide{The batch size is set to 512.} (b) The training time decreases as the number of workers increases. \model-I takes less training time to converge compared with \model-T.}
\end{figure}

\begin{figure}[t]
	\centering
	\subfigure[Base embedding dimension]{\label{fig:basedim}\includegraphics[width=0.23\textwidth]{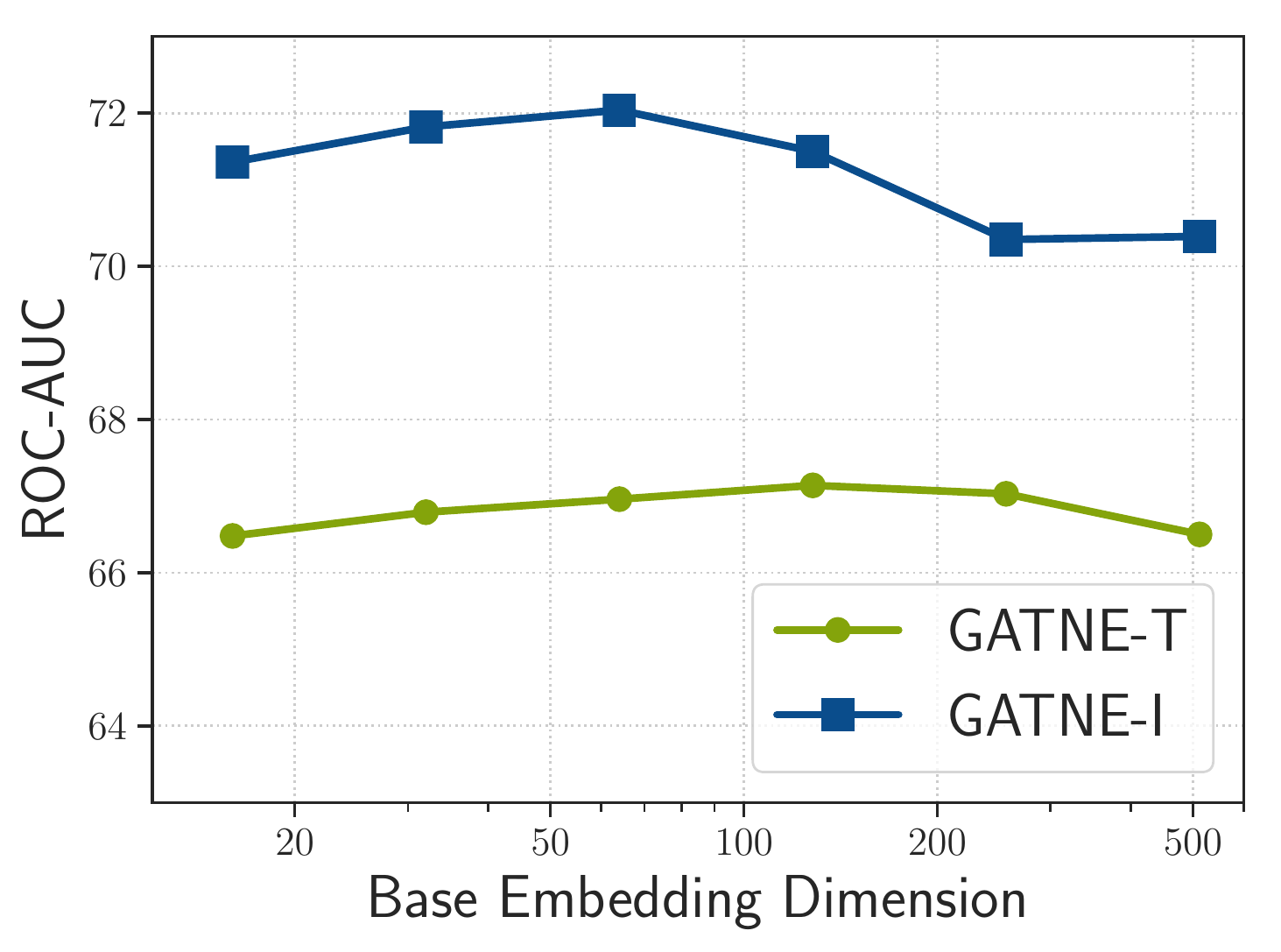}}
	\subfigure[Edge embedding dimension]{\label{fig:edgedim}\includegraphics[width=0.23\textwidth]{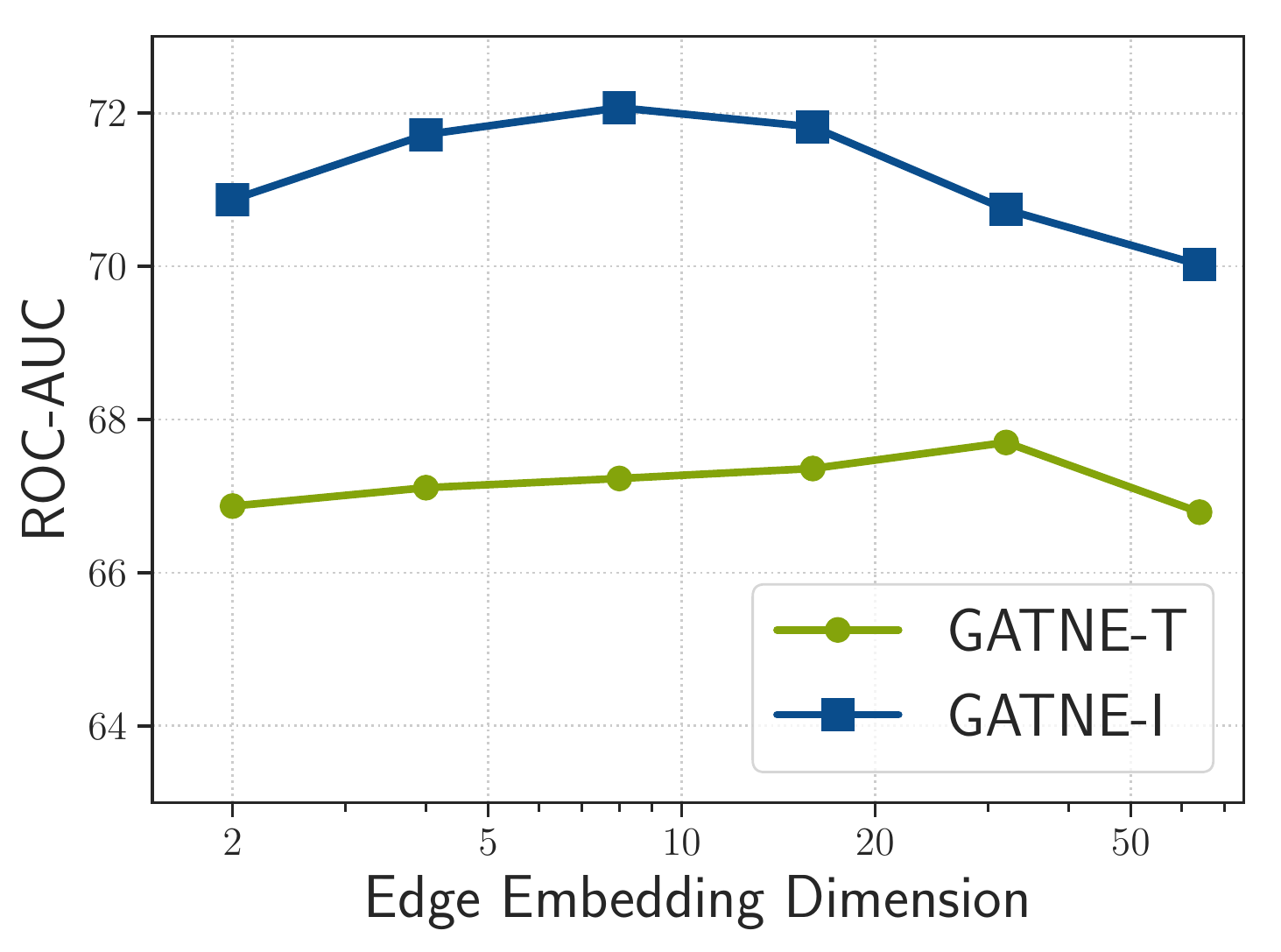}}
	\caption{The performance of \model-T and \model-I on \company-S when changing base/edge embedding dimensions exponentially.}
	\label{fig:parameter}
\end{figure}

\vpara{Quantitative Results.}
The experimental results of three public datasets and \company-S are shown in Table~\ref{tab:tab_res}. \model\ outperforms all sorts of baselines in the various datasets. 
\model-T obtains better performance than \model-I on Amazon dataset as the node attributes are limited. The node attributes of Alibaba dataset are abundant so that \model-I obtains the best performance. ANRL is very sensitive to the weak node attributes and obtains the worst result on Amazon dataset. The different node attributes of users and items also limit the performance of ANRL on Alibaba-S dataset. On YouTube and Twitter datasets, \model-I performs similarly to \model-T as the node attributes of these two datasets are the node embeddings of DeepWalk, which are generated by the network structure.
Table~\ref{tab:tab_online_res} lists the experimental results of \company\ dataset.  \model\ scales very well and achieves state-of-the-art performance on \company\  dataset, with $2.18\%$ performance lift in PR-AUC, $5.78\%$ in ROC-AUC, and $5.99\%$ in F1-score, compared with best results from previous state-of-the-art algorithms. The \model-I performs better than \model-T model in the large-scale dataset, suggesting that the inductive approach works better on large-scale attributed multiplex heterogeneous networks, which is usually the case in real-world situations.

\vpara{Convergence Analysis.}
We analyze the convergence properties of our proposed models on \company\ dataset. \hide{For detailed implementations see Appendix \ref{app:converge}.} The results, as shown in Figure~\ref{subfig:convergence}, demonstrate that \model-I converges faster and achieves better performance than \model-T on extremely large-scale real-world datasets. 

\vpara{Scalability Analysis.}
We investigate the scalability of \model\ that has been deployed on multiple workers for optimization. Figure~\ref{subfig:scalability} shows the speedup w.r.t. the number of workers on the \company\ dataset. The figure shows that \model\ is quite scalable on the distributed platform, as the training time decreases significantly when we add up the number of workers, and finally, the inductive model takes less than 2 hours to converge with 150 distributed workers. We also find that \model-I's training speed increases almost linearly as the number of workers increases but less than 150. While  \model-T converges slower and its training speed is about to reach a limit when the number of workers being larger than 100. Besides the state-of-the-art performance, \model\ is also scalable enough to be adopted in practice. 

\vpara{Parameter Sensitivity.}
We investigate the sensitivity of different hyper-parameters in \model, including overall embedding dimension $d$ and edge embedding dimension $s$.\hide{, and also the sensitivity to sparse input. In real world scenarios, usually is takes a hard effort to get full link information of graphs, but not that hard to collect some of the link information. So we also evaluate \model's performance when the link of the training graph not being provided. } Figure~\ref{fig:parameter} illustrates the performance of \model\ when altering the base embedding dimension $d$ or edge embedding dimension $s$ from the default setting ($d=200,s=10$). We can conclude that the performance of \model\ is relatively stable within a large range of base/edge embedding dimensions, and the performance drops when the base/edge embedding dimension is either too small or too large. \hide{For the implementation details see the Appendix \ref{app:para}.}





\subsection{Offline A/B Tests}
\label{sec:online}

We deploy our inductive model \model-I on \company's distributed cloud platform for its recommendation system. The training dataset has about 100 million users and 10 million items, with 10 billion interactions between them per day. We use the model to generate embedding vectors for users and items. For every user, we use K nearest neighbor (KNN) with Euclidean distance to calculate the top-N items that the user is most likely to click. The experimental goal is to maximize top-N hit-rate. 
Under the framework of A/B tests, we conduct an offline test on \model-I, MNE, and DeepWalk. The results demonstrate that \model-I improves hit-rate by \textbf{3.26\%} and \textbf{24.26\%} compared to MNE and DeepWalk, respectively.

\section{Conclusion} \label{sec:conclusion}

In our paper, we formalized the attributed multiplex heterogeneous network embedding problem and proposed \model\ to solve it with both transductive and inductive settings. We split the overall node embedding of \model-I into three parts: base embedding, edge embedding, and attribute embedding. The base embedding and attribute embedding are shared among edges of different types, while the edge embedding is computed by aggregation of neighborhood information with the self-attention mechanism. Our proposed methods achieve significantly better performances compared to previous state-of-the-art methods on link prediction tasks across multiple challenging datasets. The approach has been successfully deployed and evaluated on \company's recommendation system with excellent scalability and effectiveness.

\begin{acks}
We thank Qibin Chen, Ming Ding, Chang Zhou, and Xiaonan Fang for their comments. The work is supported by the
NSFC for Distinguished Young Scholar 
(61825602),
NSFC (61836013), 
and a research fund supported by Alibaba Group.
\end{acks}

\bibliographystyle{ACM-Reference-Format}
\bibliography{reference}

\appendix
\section{Appendix}
\label{sec:reproduce}

In the appendix, we first give the implementation notes of our proposed models. The detailed descriptions of datasets and the parameter configurations of all methods are then given. Finally, we discuss the questions about fair comparison and our future work. 

\subsection{Implementation Notes}

\vpara{Running Environment.}
The experiments in this paper can be divided into two parts. One is conducted on four datasets using a single Linux server with 4 Intel(R) Xeon(R) Platinum 8163 CPU @ 2.50GHz, 512G RAM and 8 NVIDIA Tesla V100-SXM2-16GB. The codes of our proposed models in this part are implemented with TensorFlow\footnote{\label{fn:tf}https://www.tensorflow.org/} 1.12 in Python 3.6. The other part is conducted on the full \company\ dataset using \company's distributed cloud platform which contains thousands of workers. Every two workers share an NVIDIA Tesla P100 GPU with 16GB memory. Our proposed models are implemented with TensorFlow 1.4 in Python 2.7 in this part.

\vpara{Implementation Details.}
Our codes used by single Linux server can be split into three parts: random walk, model training and evaluation. The random walk part is implemented referring to the corresponding part of DeepWalk\footnote{\label{fn:deepwalk}https://github.com/phanein/deepwalk} and metapath2vec\footnote{\label{fn:mp2vec}https://ericdongyx.github.io/metapath2vec/m2v.html}. The training part of the model is implemented referring to the word2vec part of TensorFlow tutorials\footnote{https://www.tensorflow.org/tutorials/representation/word2vec}. The evaluation part uses some metric functions from scikit-learn\footnote{https://scikit-learn.org/stable/modules/classes.html\#module-sklearn.metrics} including \textit{roc\_auc\_score}, \textit{f1\_score}, \textit{precision\_recall\_curve}, \textit{auc}. Our model parameters are updated and optimized by stochastic gradient descent with Adam updating rule~\cite{kingma2014adam}. The distributed version of our proposed models is implemented based on the coding rules of \company's distributed cloud platform in order to maximize the distribution efficiency. High-level APIs, such as \textit{tf.estimator} and \textit{tf.data}, are used for the higher coefficient of utilization of computation resources in the \company's distributed cloud platform.  

\vpara{Function Selection.}
Many different aggregator functions in Equation (\ref{aggregation}), such as the mean aggregator (Cf. Equation (\ref{mean})) or pooling aggregator (Cf. Equation (\ref{pool})), achieve similar performance in our experiments. Mean aggregator is finally used to be reported in the quantitative experiments in our model. We use the linear transformation function as the parameterized function of attributes $\mathbf{h}_z$ and $\mathbf{g}_{z,r}$ in Equation (\ref{eqn:vni}) of our inductive model \model-I.

\vpara{Parameter Configuration.}
Our base/overall embedding dimension $d$ is set to 200 and the dimension of edge embedding $s$ is set to 10. The number of walks for each node is set to 20 and the length of walks is set to 10. The window size is set to 5 for generating node contexts. The number of negative samples $L$ for each positive training sample is set to 5. The number of maximum epochs is set to 50 and our models will early stop if ROC-AUC on the validation set does not improve in 1 training epoch. The coefficient $\alpha_r$ and $\beta_r$ are all set to 1 for every edge type $r$. For \company\ dataset, the node types include $U$ and $I$ representing \textit{User} and \textit{Item} respectively. The meta-path-schemes of our methods are set to $U-I-U$ and $I-U-I$. We use the default setting of Adam optimizer in TensorFlow; the learning rate is set to $lr=0.001$. For offline A/B test in section \ref{sec:online}, we use $N=50$.

\vpara{Code and Dataset Releasing Details.}
The codes of our proposed models on the single Linux server (based on Tensorflow 1.12), together with our partition of the three public datasets and the \company-S dataset are available. 

\subsection{Compared Methods}
We give the detailed running configuration about all compared methods as follows.
\hide{We compare our model with the following methods of different types. For the experiment conducted on a single Linux server, these methods' implementation details are described as follows.} The embedding size is set to 200 for all methods. For random-walk based methods, the number of walks for each node is set to 20 and the length of walks is set to 10. The window size is set to 5 for generating node contexts. The number of negative samples for each training pairs is set to 5. The number of iterations for training the skip-gram model is set to 100. The code sources and other specific hyper-parameter settings of compared methods are explained below. 

\subsubsection{Network Embedding Methods}
\begin{itemize}
  \item \textbf{DeepWalk}~\cite{perozzi2014deepwalk}. For public and \company-S datasets, we use the codes from the corresponding author's GitHub\footref{fn:deepwalk}. For \company\ dataset, we re-implemented DeepWalk on \company\ distributed cloud platform. 
  \item \textbf{LINE}~\cite{tang2015line}. The codes of LINE are from the corresponding author's GitHub\footnote{https://github.com/tangjianpku/LINE}. We use the LINE(1st+2nd) as the overall embeddings. The embedding size is set to 100 both for first-order and second-order embeddings. The number of samples is set to 1000 million.
  \item \textbf{node2vec}~\cite{grover2016node2vec}. The codes of node2vec are from the corresponding author's GitHub\footnote{https://github.com/aditya-grover/node2vec}. Node2vec adds two parameters to control the random walk process. The parameter $p$ is set to 2 and the parameter $q$ is set to 0.5 in our experiments.
\end{itemize}

\subsubsection{Heterogeneous Network Embedding Methods}
\begin{itemize}
    \item \textbf{metapath2vec}~\cite{dong2017metapath2vec}. 
    The codes provided by the corresponding author are only for specific datasets and could not directly generalize to other datasets. We re-implement metapath2vec for networks with arbitrary node types in Python based on the original C++ codes\footnote{https://ericdongyx.github.io/metapath2vec/m2v.html}. As the number of node types of three public datasets is one, metapath2vec degrades to DeepWalk in these three datasets. For \company\ dataset, the node types include $U$ and $I$ representing \textit{User} and \textit{Item} respectively. The meta-path-schemes are set to $U-I-U$ and $I-U-I$. \hide{As the original code is tied to some certain categories and are difficult to run, we re-implemented it ourselves. }
\end{itemize}

\subsubsection{Multiplex Heterogeneous Network Embedding Methods}
\begin{itemize}
  \item \textbf{PMNE}~\cite{liu2017principled}. PMNE proposes three different methods to apply node2vec on multiplex networks. We denote their network aggregation algorithm, result aggregation algorithm, and layer co-analysis algorithm as PMNE(n), PMNE(r), and PMNE(c) respectively in accord with the denotations of MNE\cite{ijcai2018-428}. We use the codes from MNE's GitHub\footnote{\label{fn:mne}https://github.com/HKUST-KnowComp/MNE}. The probability of traversing layers of PMNE(c) is set to 0.5. 
  \item \textbf{MVE}~\cite{qu2017attention}. MVE uses collaborated context embeddings and applies an attention mechanism to view-specific embeddings. The code of MVE was received from the corresponding author by email. The embedding dimensions for each view is set to 200. The number of training samples for each epoch is set to 100 million and the number of epochs is set to 10. As for \company\ dataset, we re-implemented this method on the \company\ distributed cloud platform. 
  \item \textbf{MNE}~\cite{ijcai2018-428}. MNE uses one common embedding and several additional embeddings of each edge type for each node, which are jointly learned by a unified network embedding model. The additional embedding size for MNE is set to 10. We use the codes released by the corresponding author in the GitHub\footref{fn:mne}. As for \company\ dataset, we re-implemented it on the \company\ distributed cloud platform. 

\end{itemize}

\subsubsection{Attributed Network Embedding Methods}
\begin{itemize}
  \item \textbf{ANRL}~\cite{zhang2018anrl}. We use the codes from Alibaba's GitHub\footnote{https://github.com/cszhangzhen/ANRL}. As YouTube and Twitter datasets do not have node attributes, we generate node attributes for them. To be specific, we use the node embeddings (200 dimensions) of DeepWalk as the input node features on these datasets for ANRL. For \company-S and Amazon dataset, we use raw features as attributes. 
\end{itemize}

\begin{table}
  \centering
  \caption{\label{tab:stats_origin} Statistics of Original Datasets.}
  \begin{tabular}{c|c|c|c|c}
    \hline \hline
    \textbf{Dataset} & \# nodes & \# edges & \# n-types & \# e-types \\
    \hline
    Amazon & 312,320 & 7,500,100 & 1 & 4 \\
    YouTube & 15,088 & 13,628,895 & 1 & 5 \\
    Twitter & 456,626 & 15,367,315 & 1 & 4 \\
    \hline \hline
  \end{tabular}
\end{table}

\subsection{Datasets}
Our experiment evaluates on five datasets, including four datasets and \company\ dataset. Due to the limitation of memory and computation resources on a single Linux server, the four datasets are subgraphs sampled from the original datasets for training and evaluation. Table \ref{tab:stats_origin} shows the statistics of the original public datasets. 

\begin{itemize}
    \item \textbf{Amazon} is a dataset of product reviews and metadata from Amazon. In our experiments, we only use the product metadata, including the product attributes and co-viewing, co-purchasing links between products. The node type set of Amazon is $\mathcal{O}=\{product\}$ and the edge type set of Amazon is $\mathcal{R}=\{also\_bought, also\_viewed\}$, which denotes two products are co-bought or co-viewed by the same user respectively. The products of Amazon are split into many categories. The number of products in all the categories is so large that we use the \textit{Electronics} category of products for experimentation. The number of products in \textit{Electronics} is still large for many algorithms; therefore, we extract a connected subgraph from the whole graph. 
    \item \textbf{YouTube} is a multi-dimensional bidirectional network dataset consists of 5 types of interactions between 15,088 YouTube users. The types of edges include contact, shared friends, shared subscription, shared subscriber, and shared favorite videos between users. It is a multiplex network with $|\mathcal{O}|=1$ and $|\mathcal{R}|=5$.
    \item \textbf{Twitter} is a dataset about tweets posted on Twitter about the discovery of the Higgs boson between 1st and 7th, July 2012. It is made up of $4$ directional relationships between more than 450,000 Twitter users. The relationships are re-tweet, reply, mention, and friendship/follower relationship between Twitter users. It is a multiplex network with $|\mathcal{O}|=1$ and $|\mathcal{R}|=4$.
    \item \textbf{\company} consists of $4$ types of interactions which including click, add-to-preference, add-to-cart, and conversion between two types of nodes, user and item. The node type set of \company\ is $\mathcal{O}=\{user, item\}$ and the size of the edge type set is $|\mathcal{R}|=4$. The whole graph of \company\ is so large that we cannot evaluate the performances of different methods on it by a single machine. We extract a subgraph from the whole graph for comparison with different methods, denoted as \textbf{\company-S}. By the way, we also provide the evaluation of the whole graph on the \company's distributed cloud platform, the full graph is denoted as \textbf{\company}.
\end{itemize}

\subsection{Discussion}

As for research on network embedding, many people use link prediction or node classification tasks for evaluating the representation of network embeddings. However, although there are many commonly used public datasets, like Twitter or YouTube dataset, none of them provide a "standard" separation for train, validation, and test for different tasks. This causes different results on the same dataset for different evaluation separations so the results from previous papers cannot be directly used, and researchers have to re-implement and run all baselines themselves, reducing their attention on improving their model.

Here we appeal on researchers to provide the standardized dataset, which contains a standard separation of train, validation and test sets as well as the full dataset. Therefore researchers can evaluate their method based on a standard environment, and results across papers can be compared directly. This also helps to increase the reproducibility of research. 

\vpara{Future Work.}
Apart from the heterogeneity of networks, the dynamics of networks are also crucial to network representation learning. There are three ways to capture the dynamic information of networks. Firstly, we can add dynamic information into node attributes. For example, we can use methods like LSTM~\cite{hochreiter1997long} to capture the dynamic activities of users. Secondly, the dynamic information, such as the timestamp of each interaction, can be considered as the attributes of edges. Thirdly, we may consider the several snapshots of networks representing the dynamic evolution of networks. We leave representation learning for the dynamic attributed multiplex heterogeneous network as our future work.



\end{document}